%% file: main.tex
\documentclass[journal,twoside,web]{ieeecolor}
\usepackage{lcsys}
\usepackage{amsmath}
\usepackage{amssymb}
\usepackage{mathrsfs}

\usepackage{verbatim}
\let\labelindent\relax  % Disable \labelindent to use enumitem package
\usepackage{enumitem}
\usepackage{graphicx}
\usepackage[colorlinks]{hyperref}

\hypersetup{
  urlcolor      = green,
  linkcolor     = blue,
  citecolor     = red
}

\newtheorem{assumption}{Assumption}
\newtheorem{theorem}{Theorem}
\newtheorem{lemma}{Lemma}
\newtheorem{proposition}{Proposition}

\newtheorem{definition}{Definition}

\newtheorem{remark}{Remark}

\input{commands}

\title{On Regular Regressors in Adaptive Control}
\author{
    Erick {Mejia Uzeda},
    \IEEEmembership{Graduate Student Member, IEEE},
    and 
    Mireille E. Broucke,
    \IEEEmembership{Member, IEEE}
    \thanks{This work was supported by the Natural Sciences and Engineering Research Council of Canada (NSERC).}
    \thanks{The authors are with the Department of Electrical and Computer Engineering, University of Toronto, Toronto, ON M5S 3G4, Canada (e-mail: erick.mejiauzeda@mail.utoronto.ca; broucke@control.utoronto.ca).}
}   

\begin{document}

\maketitle

% % Add page numbers
% \thispagestyle{plain}
% \pagestyle{plain}

% Remove headers + page numbers
\pagestyle{empty}
\thispagestyle{empty}

\begin{abstract}
This paper addresses a shortcoming in adaptive control, that the property of a regressor being persistently exciting (PE) is not well-behaved. One can construct regressors that upend the commonsense notion that excitation should not be created out of nothing. To amend the situation, a notion of regularity of regressors is needed.
We are naturally led to a broad class of regular regressors that enjoy the property that their excitation is always confined to a subspace, a foundational result called the PE decomposition. A geometric characterization of regressor excitation opens up new avenues for adaptive control, as we demonstrate by formulating a number of new adaptive control problems.
\end{abstract}


\section{Introduction}

\IEEEPARstart{A}{} regressor $w(t) \in \RR^q$ is a signal used to reconstruct other unknown signals $d(t) = w\T(t) \psi \in \RR$ or to identify unknown constant parameters $\psi \in \RR^q$. 
In adaptive control the goal is to design an adaptive law for $\psih(t) \in \RR^q$ so that an error, such as $e := w\T \psih - d$, vanishes asymptotically.
Convergence of $\psih$ to $\psi$ is known to be attainable when a condition called {\em persistent excitation} (PE) is satisfied by the regressor.

\begin{definition}%[Persistent Excitation] 
    \label{def:PE}
    We say $w(t) \in \RR^q$ is {\em persistently exciting} (PE) if there exist constants $\beta\pe,\, T\pe > 0$ such that
    \begin{align} \label{eq:PE}
        \frac{1}{T\pe} \int_{t}^{t + T\pe} w(\tau) w\T(\tau) \,d\tau
        \succeq
        \beta\pe I
    \end{align}
    for all $t \geq 0$.
    \rqed
\end{definition}

Historically, definitions for PE were developed to ensure global uniform asymptotic stability (GUAS) of systems commonly arising in adaptive control \cite{NARENDRA89, PANTELEY01}. 
Other notions of PE, relevant to the behavioural approach to control, can be found in \cite{WILLEMS05, NESIC15, TESI20, MARKOVSKY21} and are not the focus of this paper.
In \cite{NARENDRA89} the focus is on the classic adaptive law
\begin{align*}
    \dot{\psih}
    &=
    - \Gamma w(t) e
\end{align*}
where $\Gamma \succ 0$ and $e(t) \in \RR$ is the error produced from an error model.
Definition~\ref{def:PE} is then shown to be a necessary and sufficient condition for GUAS of the {\em static error model}
\begin{align*}
    e &= w\T(t) (\psih - \psi)
\end{align*}
but not for the {\em dynamic error model}
\begin{align*}
    \dot{x} &= A x + B w\T(t) (\psih - \psi)
    \\
    e &= B\T P x
    \,.
\end{align*}
Instead, a suitable definition of PE in the dynamic error model case is found in \cite[p.~247]{NARENDRA89}.
However, if one adopts the class of regressors proposed in \cite{YUAN77}, denoted $\cP_{[0, \infty)}$, that are piecewise continuous with bounded derivative, then the two definitions of PE are equivalent \cite[Corollary~2.3]{NARENDRA89}. The requirement that a regressor belongs to $\cP_{[0, \infty)}$ is interpreted as a {\em regularity condition}, which has been widely adopted in the literature \cite{NARENDRA89,YUAN77,SASTRY89,IOANNOU12}.

A significant contribution to the understanding of the PE condition was provided by \cite{NARENDRA87}. It introduced a notion of PE in a subspace, determined how PE is modified under various transformations, and made connections to robustness. The spirit of the present paper is to further the program of \cite{NARENDRA87}. 
Specifically, we identify an important shortcoming in adaptive control, that the PE condition is not well-behaved, even for regressors in $\cP_{[0, \infty)}$. This demonstrates that classical regularity conditions need to be revisited.

We start from the characterization of PE that requires regressors to be PE along any direction in the parameter space. 

\begin{proposition}[{\cite[Lemma~2]{NARENDRA87}}]
    \label{prop:PE:PE_al}
    Suppose $w(t) \in \RR^q$ is bounded.
    Then $w$ is PE if and only if $\al\T w$ is PE for all non-zero $\al \in \RR^q$.
    \dqed
\end{proposition}

This characterization of PE is important because it is formulated in the context in which regressors are used in adaptive control: as an inner product between the regressor and the unknown parameter vector. It also provides insight on where in the parameter space unknown parameters can be recovered. Proposition~\ref{prop:PE:PE_al} is the foundation from which we carry out a deeper study of the PE condition.

Our contributions are as follows.
In Section~\ref{sec:reg} we elaborate on what it means for the PE condition to be well-behaved, defining a new notion of regularity involving the so-called non-PE set. A pathological regressor inspired from \cite{NARENDRA89} is shown to be non-regular.
In Section~\ref{sec:PEdecomp} we present the PE decomposition, shedding light on the geometric nature of the PE condition. It shows that a notion of PE in a subspace \cite{SADEGH_HOROWITZ90, LI_KRSTIC98, MARINOTOMEI23, GOEL20, GUO23, ARANOVSKI23} does not need to be explicitly imposed, but rather it is an intrinsic property of regular regressors.
Our geometric approach to PE is different from the approach in \cite{PADOAN17}, which focuses on geometric properties of dynamical systems.
In Section~\ref{sec:conds} we reconcile regularity with the the notion that regressor components that are not PE should be vanishing.
Finally, new problems in adaptive control are introduced in Section~\ref{sec:probs}.

\section{Regularity of Regressors}
\label{sec:reg}

Classical regularity conditions focus on continuity and differentiability properties of a regressor \cite{YUAN77}.
We will see that such conditions do not capture the key property required for the PE condition to be well-behaved. As a result, a new and carefully crafted notion of regularity is needed.

\subsection{Regularity and the Non-PE Set}
\label{sec:PE:reg:nonPE}

To illustrate what we want from a regressor so that the PE condition is well-behaved, suppose that we have a regressor $w(t) \in \RR^4$ where the components $w_1,\, w_2$ are each PE and the components $w_3,\, w_4$ are each not PE. Let $\al_i \neq 0$.
There are four {\em reasonable} behaviours we expect to see when taking linear combinations of scalar components:
\begin{enumerate}[label=\alph*)]
    \item \label{itm:PE:reg:a}
    $\al_1 w_1 + \al_2 w_2$ may be PE;
    \item \label{itm:PE:reg:b}
    $\al_1 w_1 + \al_2 w_2$ may be not PE;
    \item \label{itm:PE:reg:c}
    $\al_2 w_2 + \al_4 w_4$ is PE;
    \item \label{itm:PE:reg:d}
    $\al_3 w_3 + \al_4 w_4$ is not PE.
\end{enumerate}
Behaviour \ref{itm:PE:reg:a} captures the notion of linear independence of PE components of a regressor.
Behaviour \ref{itm:PE:reg:b} instead captures the notion of linear dependence of PE components of a regressor.
One of these two behaviours always holds true. 
Behaviours \ref{itm:PE:reg:c} and \ref{itm:PE:reg:d} state that the addition of a scalar component that is not PE does not alter the excitation properties of other components. 
These observations provide a roadmap for a new regularity condition. First, we define a notion complementing the characterization of PE in Proposition~\ref{prop:PE:PE_al}.

\begin{definition}%[Non-PE] 
    \label{def:PE:nonPE}
    We say $w(t) \in \RR^q$ is {\em non-PE} if for all $\al \in \RR^q$ we have $\al\T w$ is not PE.
    \rqed
\end{definition}

Notice that a scalar regressor $w(t) \in \RR$ is not PE if and only if it is non-PE.
Now to address if behaviours \ref{itm:PE:reg:c} and \ref{itm:PE:reg:d} can be achieved, we define the {\em non-PE set of $w(t) \in \RR^q$} as
\begin{align} \label{eq:PE:nonPEset}
    \cW^{\star} 
    :=
    \{\, \al \in \RR^q ~:~ \al\T w \text{ is not PE} \,\}
    \,.
\end{align}
From Definition~\ref{def:PE}, this set is closed under scalar multiplication. Therefore, $\cW^{\star}$ is at the very least a {\em pencil} of lines. For behaviours \ref{itm:PE:reg:c} and \ref{itm:PE:reg:d} to hold true, $\cW^{\star}$ should also be closed under vector addition; that is, $\cW^{\star}$ should be a subspace.
If behaviours \ref{itm:PE:reg:c} and \ref{itm:PE:reg:d} do not hold true, then it would be possible to create excitation out of signals without excitation, demonstrating that the PE condition is not well-behaved.
We arrive at a notion of a regular regressor, cast in terms of geometric structure, characterizing a well-behaved regressor.

\begin{definition}%[Regular Regressor] 
    \label{def:PE:reg}
    We say $w(t) \in \RR^q$ is a {\em regular regressor} if it is bounded and its non-PE set is a subspace.
    We call $\cW^{\star}$ the {\em non-PE subspace of $w$}.
    \rqed
\end{definition}

Notice this notion of regularity does not at the outset impose continuity nor differentiability (even if such assumptions are useful to consolidate PE definitions). In that sense, we move beyond classical notions of regularity. For example, a regressor $w(t) \in \RR^q$ whose components are each the Dirichlet function is nowhere continuous yet satisfies Definition~\ref{def:PE:reg} since its non-PE set is $\RR^q$ by the fact that it is zero almost everywhere. On the other hand, some form of continuity of regressors is a natural assumption that we will revisit as we further characterize regular regressors. 

\subsection{Pathological Regressors}
\label{sec:PE:reg:pathology}

This section provides an example capturing the essence of the pathology that leads to a violation of the behaviours \ref{itm:PE:reg:c} or \ref{itm:PE:reg:d}. This pathology is known in the literature \cite[p. 239]{NARENDRA89} but, to our knowledge, was not further investigated. We firstly point out that any scalar regressor is always a regular regressor. This is because the non-PE set is closed under scalar multiplication and any subspace of $\RR$ is either $\{\, 0 \,\}$ or $\RR$.
We construct a non-regular regressor $w(t) \in \RR^2$ whose components $w_i$ are not PE, but their sum $\al\T w = w_1 + w_2$ is PE. The key idea is found in the following result.

\begin{proposition} \label{prop:PE:pathology}
    Suppose $v(t) \in \RR^q$ is PE.
    Let $\gamma(t) \in \RR$ be constructed as follows: partition $[0, \infty)$ into intervals of length $2^k$ for $k \in \NN$, set $\gamma(t) = 1$ for every $k$-th interval where $k$ is odd, and set $\gamma(t) = 0$ elsewhere.
    Then $w_1 := \gamma v$ and $w_2 := (1 - \gamma) v$ are each non-PE.
    \dqed
\end{proposition}
\begin{proof}
    To show that $w_1$ is non-PE, we need to show that $\al\T w_1$ is not PE for all $\al \in \RR^q$.
    Fix any $\beta,\, T > 0$. Let $N \in \NN$ be such that $T \leq 2^{N}$ and $N$ is even. Define $t_0 := \sum_{k = 1}^{N - 1} 2^k$ and $t_1 := t_0 + 2^N$. Then we have that $\gamma(t) = 0$ for all $t \in (t_0, t_1)$ by construction of $\gamma$, and so
    \begin{align*}
        \frac{1}{T} \int_{t_0}^{t_0 + T} w_1(\tau) w_1\T(\tau) \,d\tau
        &=
        \frac{1}{T} \int_{t_0}^{t_0 + T} \gamma^2(\tau) v(\tau) v\T(\tau) \,d\tau
        \\&=
        0
    \end{align*}
    since $t_0 + T \leq t_0 + 2^N = t_1$. As a result, for any $\al$ we have
    \begin{align*}
        \frac{1}{T} \int_{t_0}^{t_0 + T} \al\T w_1(\tau) w_1\T(\tau) \al \,d\tau
        =
        0
        <
        \beta \cdot 1
        \,.
    \end{align*}
    Since the choice of $\beta,\, T > 0$ was arbitrary, it cannot be that $\al\T w_1$ is PE, and so the scalar regressor $\al\T w_1$ is not PE.
    Since $\al$ was also arbitrary, this means that $w_1$ is non-PE.

    The proof that $w_2$ is non-PE is similar, hinging on the fact that $1 - \gamma$ looks like $\gamma$ except the intervals where $1 - \gamma$ are $0$ versus $1$ are swapped. As such, the only modification to the proof is to select $N$ odd.
\end{proof}

Proposition~\ref{prop:PE:pathology} shows the existence of regressors violating behaviours \ref{itm:PE:reg:c} and \ref{itm:PE:reg:d}. Particularly, 
$v - w_2 = w_1$ violates behaviour \ref{itm:PE:reg:c} while
$w_1 + w_2 = v$ violates behaviour \ref{itm:PE:reg:d}.
From here we can construct a simple but evocative regressor that is not regular. Let $v(t) = 1$ and define $w(t) := (w_1, w_2)(t) = (\gamma, 1 - \gamma)(t) \in \RR^2$.
Since both $w_1$ and $w_2$ are not PE, the non-PE set of $w$ contains  $(\RR \x \{\, 0 \,\}) \cup (\{\, 0 \,\} \x \RR)$.
Let $\al = (\al_1, \al_2) \in \RR^2$ with $\al_i \neq 0$; that is, $\al \not\in (\RR \x \{\, 0 \,\}) \cup (\{\, 0 \,\} \x \RR)$. Then $| \al\T w | \geq \min\{\, | \al_1 |, | \al_2 | \,\}$ because $\al\T w = \al_1 w_1 + \al_2 w_2$ alternates between $\al_1$ and $\al_2$ every interval of length $2^k$. As a result, $\al\T w$ is PE,  so $\al \not\in \cW^{\star}$.
Hence $\cW^{\star} = (\RR \x \{\, 0 \,\}) \cup (\{\, 0 \,\} \x \RR)$ is a cross made up of the Cartesian axes, meaning it is not a subspace. Therefore $w$ is not a regular regressor. We additionally point out that since the components of $w$ are pulse trains, $w \in \cP_{[0, \infty)}$ and its derivative is zero almost everywhere. This example therefore highlights the shortcomings of classical regularity conditions focused on continuity properties.

\subsection{Properties of Regular Regressors}

This section shows that regular regressors enjoy a number of properties useful in analysis. The proofs themselves demonstrate the ease of working with a geometric notion of regular regressors.
We first present a property about non-PE regressors that is independent of regularity.

\begin{lemma} \label{lem:noPE:L}
    Consider the matrix $L \in \RR^{p \x q}$.
    If $w(t) \in \RR^q$ is non-PE then $L w$ is non-PE.
    \dqed
\end{lemma}
\begin{proof}
    Pick any $\al \in \RR^p$ and notice $L\T \al \in \RR^q$. Since $w$ is non-PE, we have that $(L\T \al)\T w = \al\T L w$ is not PE. Given that $\al$ was arbitrary, this means $L w$ is non-PE.
\end{proof}

Next we show regularity is preserved under linear maps.

\begin{lemma} \label{lem:PE:reg:component}
    Suppose $w(t) \in \RR^q$ is a regular regressor.
    Let $L \in \RR^{p \x q}$ be a matrix.
    Then $L w$ is a regular regressor.
    \dqed
\end{lemma}
\begin{proof}
    Let $\cW^{\star}$ denote the non-PE subspace of $w$ and let $\cW_L^{\star}$ denote the non-PE set of $L w$. Since the non-PE set is always closed under scalar multiplication, it suffices to show that it is closed under addition.
    Let $\al_1,\, \al_2 \in \cW_L^{\star}$. Then $\al_i\T L w = (L\T \al_i)\T w$ is not PE and so $L\T \al_i \in \cW^{\star}$. Because $\cW^{\star}$ is a subspace, we have that $L\T \al_1 + L\T \al_2 \in \cW^{\star}$. But then $(L\T \al_1 + L\T \al_2)\T w = (\al_1 + \al_2)\T L w$ is not PE, meaning that $\al_1 + \al_2 \in \cW_L^{\star}$.
\end{proof}

Finally, we present some facts about combinations of regressor components, generalizing \cite[Lemma~6.5~(ii)]{NARENDRA89}.  Theorem~\ref{thm:PE:stack:props} shows that regular regressors are well-behaved according to the behaviours \ref{itm:PE:reg:a}-\ref{itm:PE:reg:d} of Section~\ref{sec:PE:reg:nonPE}.

\begin{theorem} \label{thm:PE:stack:props}
    Suppose $w(t) := (w_1, w_2)(t) \in \RR^{q_1} \x \RR^{q_2}$ is a regular regressor and $w_1$ is non-PE.
    Then:
    \begin{enumerate}[label=\alph*)]
        \item \label{itm:PE:stack:props:a}
        if $w_2$ is non-PE then $w$ is non-PE;
        \item \label{itm:PE:stack:props:b}
        if $q_1 = q_2$ and $w_2$ is non-PE then $w_1 + w_2$ is non-PE;
        \item \label{itm:PE:stack:props:c}
        if $q_1 = q_2$ and $w_2$ is PE then $w_1 + w_2$ is PE.
        \dqed
    \end{enumerate}
\end{theorem}
\begin{proof}
    Throughout we have that $w_1$ is non-PE.
    
    \ref{itm:PE:stack:props:a}
    Suppose $w_2$ is non-PE. 
    Let $\al = (\al_1, \al_2) \in \RR^{q_1} \x \RR^{q_2}$ and let $\cW^{\star}$ denote the non-PE subspace of $w$. Since each $\al_i\T w_i$ is not PE, because $w_i$ is non-PE, we have that $(\al_1, 0),\, (0, \al_2) \in \cW^{\star}$. But $\cW^{\star}$ is a subspace, since $w$ is a regular regressor, and so $\al = (\al_1, 0) + (0, \al_2) \in \cW^{\star}$.

   \ref{itm:PE:stack:props:b}
    Suppose $w_2$ is non-PE. 
    By \ref{itm:PE:stack:props:a} we know $w$ is non-PE and the result follows by Lemma~\ref{lem:noPE:L} since $w_1 + w_2 = \begin{bmatrix}I & I\end{bmatrix} w$.

    \ref{itm:PE:stack:props:c}
    Suppose $w_2$ is PE. 
    For the sake of contradiction, suppose that $w_1 + w_2$ is not PE. Then there exists a non-zero $\al \in \RR^q$ such that $\al\T (w_1 + w_2)$ is not PE. Since $\al\T (w_1 + w_2)$ is a scalar, it is also non-PE. Since $w_1$ is non-PE, we have that $- \al\T w_1$ is also non-PE.
    By Lemma~\ref{lem:PE:reg:component} we have that $(\al\T (w_1 + w_2), - \al\T w_1)$ is a regular regressor. By \ref{itm:PE:stack:props:b} we have that $\al\T (w_1 + w_2) + (- \al\T w_1) = \al\T w_2$ is not PE. But $\al$ is non-zero, contradicting the fact that $w_2$ is PE according to Proposition~\ref{prop:PE:PE_al}.
\end{proof}

\section{The PE Decomposition}
\label{sec:PEdecomp}

The previous section identified several advantageous properties of regular regressors. However, the most important property is that any available excitation of a regular regressor is always confined to a subspace of the regressor space called the {\em PE subspace}. Moreover, there is always a component of the regressor (a reduced order regressor) that manifests all the excitation properties of the original regressor.

\begin{definition} \label{def:PEsub}
    Let $\cW^{\star}$ denote the non-PE set of $w$.
    The {\em PE subspace of $w$} is the subspace $\cW := (\cW^{\star})^{\perp}$. 
    We say the {\em degree of PE of $w$} is $q\pe := \dim(\cW)$.
    \rqed
\end{definition}

Unlike the non-PE set $\cW^{\star}$, the set $\cW$ is always a subspace (even for non-regular regressors) since the standard inner product is linear in each argument.
For a regular regressor we have that $\cW^{\star} = \cW^{\perp}$ and so we will generally denote $\cW^{\perp}$ as the non-PE subspace (instead of $\cW^{\star}$).
 
We present a decomposition of regular regressors characterizing their inherent excitation. This result is based on projection maps that have been widely used in linear geometric control \cite{WONHAM85}.
Given subspaces $\cV_i$ satisfying $\cV_1 \oplus \cV_2 = \RR^q$, let $\Pi_{\cV_1}^{\cV_2} \in \RR^{q \x q}$ denote the {\em projection on $\cV_1$ along $\cV_2$}.
Let $U_{\cV_1}$ denote an {\em insertion map of $\cV_1$}, meaning it has full column rank and $\Img(U_{\cV_1}) = \cV_1$.
Let $D_{\cV_2}$ denote a {\em natural projection along $\cV_2$}, meaning $D_{\cV_2}\T$ is an insertion map of $\cV_2^{\perp}$. In other words, $D_{\cV_2}$ has full row rank and $\Ker(D_{\cV_2}) = \cV_2$.
We specifically consider pairs $(U_{\cV_1}, D_{\cV_2})$ acting as a coordinate transformation between $\cV_1$ and $\RR^{\dim(\cV_1)}$.
For this to be the case, we must impose that $D_{\cV_2} U_{\cV_1} = I$ and we call such a pair a {\em projection pair}. The proof of the following relies on showing $U_{\cV_1} D_{\cV_2}$ is idempotent \cite[p.~9]{WONHAM85}.

\begin{proposition}
    Suppose $(U_{\cV_1}, D_{\cV_2})$ is a projection pair.
    Then $U_{\cV_1} D_{\cV_2} = \Pi_{\cV_1}^{\cV_2}$.
    \dqed
\end{proposition}
% \begin{proof}
%     Given that $U_{\cV_1}$ and $D_{\cV_2}$ are full rank, it follows that $\Img(U_{\cV_1} D_{\cV_2}) = \cV_1$ and $\Ker(U_{\cV_1} D_{\cV_2}) = \cV_2$. Therefore, it suffices to show that $U_{\cV_1} D_{\cV_2}$ is idempotent \cite[p.~9]{WONHAM85}. By the fact that $D_{\cV_2} U_{\cV_1} = I$, we have
%     \begin{align*}
%         (U_{\cV_1} D_{\cV_2})^2
%         &=
%         U_{\cV_1} (D_{\cV_2} U_{\cV_1}) D_{\cV_2}
%         =
%         U_{\cV_1} D_{\cV_2}
%     \end{align*}
%     as desired.
% \end{proof}

We recall a useful fact concerning the orthogonal complement of a sum of subspaces \cite[p.~22]{WONHAM85}.

\begin{proposition} \label{prop:LA:perp_sum}
    Let $\cV_1,\, \cV_2 \subseteq \RR^q$ be subspaces. Then $(\cV_1 + \cV_2)^{\perp} = \cV_1^{\perp} \cap \cV_2^{\perp}$.
    \dqed
\end{proposition}

Now consider the case when $\cV_1 = \cW$ is the PE subspace and $\cV_2 = \cV$ satisfies $\cW \oplus \cV = \RR^q$. Based on the foregoing, the projection pairs $(U_{\cW}, D_{\cV})$ and $(U_{\cV}, D_{\cW})$ define a change of coordinates on $\RR^q$ satisfying
\begin{alignat*}{3}
    \Pi_{\cW}^{\cV} &= U_{\cW} D_{\cV}
    \,,&\quad
    I &= D_{\cV} U_{\cW}
    \,,&\quad
    0 &= D_{\cW} U_{\cW}
    \,,
    \\
    \Pi_{\cV}^{\cW} &= U_{\cV} D_{\cW}
    \,,&\quad
    I &= D_{\cW} U_{\cV}
    \,,&\quad
    0 &= D_{\cV} U_{\cV}
    \,.
\end{alignat*}

\begin{theorem}[PE Decomposition] \label{thm:PEdecomp}
    Suppose $w(t) \in \RR^q$ is a regular regressor.
    Let $\cW$ denote the PE subspace of $w$ and let $q\pe$ denote its degree of PE. Then for any $\cV$ such that $\cW \oplus \cV = \RR^q$ we have that
    \begin{align*}
        w 
        &=
        \Pi_{\cW}^{\cV} w + \Pi_{\cV}^{\cW} w
        = 
        U_{\cW} D_{\cV} w + U_{\cV} D_{\cW} w
        \\&
        =:
        U_{\cW} w\pe + U_{\cV} w\pr
    \end{align*}
    with $w\pe(t) \in \RR^{q\pe}$ PE and $w\pr(t) \in \RR^{q - q\pe}$ non-PE.
    \dqed
\end{theorem}
\begin{proof}
    First we show that $w\pe$ is PE.
    Suppose not.
    Then there exists a non-zero $\al\pe \in \RR^{q\pe}$ such that $\al\pe\T w\pe$ is not PE. Define $\al := D_{\cV}\T \al\pe \in \RR^q$ and observe that 
    \begin{align*}
        \al\T w
        &=
        \al\pe\T D_{\cV} U_{\cW} w\pe + \al\pe\T D_{\cV} U_{\cV} w\pr
        % \\&
        % =
        % \al\pe\T I w\pe + \al\pe\T 0 w\pr
        % \\&
        =
        \al\pe\T w\pe
    \end{align*}
    since $(U_{\cW}, D_{\cV})$ and $(U_{\cV}, D_{\cW})$ are projection pairs.
    Also, since $\al\T w = \al\pe\T w\pe$ is not PE by assumption, we have that $\al \in \cW^{\star} = \cW^{\perp}$. Given that $\Img(D_{\cV}\T) = \cV^{\perp}$ by definition of the natural projection, we have $\al \in \cW^{\perp} \cap \cV^{\perp}$.
    By Proposition~\ref{prop:LA:perp_sum} we have $\cW^{\perp} \cap \cV^{\perp} = (\cW + \cV)^{\perp} = (\RR^q)^{\perp} = \{\, 0 \,\}$ and so $\al = 0$. But then $\al\pe \in \Ker(D_{\cV}\T) = \{\, 0 \,\}$, given that the natural projection has full row rank, contradicting the fact that $\al\pe \neq 0$.

    Next we show that $w\pr$ is non-PE.
    For any $\al\pr \in \RR^{q - q\pe}$ define $\al := D_{\cW}\T \al\pr \in \RR^q$. Given that $\Img(D_{\cW}\T) = \cW^{\perp}$ we have $\al \in \cW^{\perp}$. Since $\al \in \cW^{\perp} = \cW^{\star}$, we know that $\al\T w$ is not PE. Proceeding similarly to our work above, one can show that $\al\T w = \al\pr\T w\pr$, and so $\al\pr\T w\pr$ is not PE. Since $\al\pr$ was arbitrary, this proves the result.
\end{proof}

% \begin{remark}
%     The regressor $0 \in \RR^0$ should be viewed both as PE and non-PE.
%     \tqed
% \end{remark}

The proof shows that the PE decomposition is a direct consequence of a geometric notion of regularity in terms of the subspace 
$\cW^{\star} = \cW^{\perp}$. The excitation of the regressor is captured by a reduced order regressor $w\pe$. This finding is related to previous work on {\em partial PE} or {\em semi-PE} \cite{SADEGH_HOROWITZ90, LI_KRSTIC98, MARINOTOMEI23, GOEL20, GUO23, ARANOVSKI23}. The key distinction is that the PE decomposition emerges as an inherent property of a regular regressor, whereas prior works take partial PE as the starting point. As such, regularity may be viewed as the more fundamental property.

The idea of finding a reduced order regressor to characterize the PE of the original regressor can be extended.
\begin{lemma} \label{lem:PE:L:PEsub}
    Suppose $w(t) \in \RR^q$ is a regular regressor.
    Let $\cW$ denote the PE subspace of $w$ and let $L \in \RR^{p \x q}$ be a matrix.
    Then the PE subspace of $L w$ is $L \cW$.
    \dqed
\end{lemma}
\begin{proof}
    Let $\cW_{L}$ denote the PE subspace of $L w$.
    Since $\cW$ is a subspace by regularity, it suffices to show that $\cW_{L}^{\star} = (L \cW)^{\perp}$.
    Notice that $\al \in \cW_{L}^{\star}$ if and only if $\al\T L w$ is not PE, meaning that $L\T \al \in \cW^{\star}$ where $\cW^{\star}$ is the non-PE set of $w$. By regularity we know that $\cW^{\star} = \cW^{\perp}$ and so the above is equivalent to the fact that for all $v \in \cW$ we have that $\al\T L v = 0$. Since $v \in \cW$ is arbitrary, this is the same as saying that $L v \in L \cW$ is arbitrary; that is, $\al \in (L \cW)^{\perp}$.
\end{proof}

Finally, we conclude this section with the significant finding that the only subspace satisfying the conclusions of Theorem~\ref{thm:PEdecomp} is the PE subspace; in other words, no other subspace characterizes the excitation of a regressor.

\begin{theorem}
    Suppose $w(t) \in \RR^q$ is a regular regressor.
    Let $\cW$ denote the PE subspace of $w$ and let $\cW\0 \subseteq \RR^q$ be a subspace with the property that for any $\cV\0$ such that $\cW\0 \oplus \cV\0 = \RR^q$ we have that
    \begin{align*}
        w
        =
        U_{\cW\0} D_{\cV\0} w + U_{\cV\0} D_{\cW\0} w
        =:
        U_{\cW\0} w_1 + U_{\cV\0} w_2
    \end{align*}
    with $w_1$ PE and $w_2$ non-PE.
    Then $\cW\0 = \cW$.
    \dqed
\end{theorem}
\begin{proof}
    Using the PE decomposition, we have two representations of the regressor $w$. They are
    \begin{align*}
        w
        &=
        U_{\cW} w\pe + U_{\cV} w\pr
        \,, \quad
        w
        =
        U_{\cW\0} w_1 + U_{\cV\0} w_2
        \,.
    \end{align*}
    To prove the result, we will compare the mismatch between these representations.
    Since $\Pi_{\cV}^{\cW} w = U_{\cV} w\pr$ and $\Pi_{\cV\0}^{\cW\0} w = U_{\cV\0} w_2$, they are each non-PE by Lemma~\ref{lem:noPE:L}. 
    Subtracting the representation resulting from $\cW\0 \oplus \cV\0 = \RR^q$ from the PE decomposition, we have that
    \begin{align*}
        0
        &=
        w - w
        =
        (\Pi_{\cW}^{\cV} - \Pi_{\cW\0}^{\cV\0}) w 
        + \Pi_{\cV}^{\cW} w + (- \Pi_{\cV\0}^{\cW\0} w)
        \,.
    \end{align*}
    Notice that $(\Pi_{\cW}^{\cV} - \Pi_{\cW\0}^{\cV\0}) w$ is non-PE by Theorem~\ref{thm:PE:stack:props}, since all the other terms are non-PE.
    Then by the PE decomposition we obtain
    \begin{align*}
        (\Pi_{\cW}^{\cV} - \Pi_{\cW\0}^{\cV\0}) w
        &=
        (\Pi_{\cW}^{\cV} - \Pi_{\cW\0}^{\cV\0}) (\Pi_{\cW}^{\cV} + \Pi_{\cV}^{\cW}) w
        \\&
        =:
        \wt
        + (\Pi_{\cW}^{\cV} - \Pi_{\cW\0}^{\cV\0}) \Pi_{\cV}^{\cW} w
        \,.
    \end{align*}
    Since $(\Pi_{\cW}^{\cV} - \Pi_{\cW\0}^{\cV\0}) \Pi_{\cV}^{\cW} w$ is also non-PE by Lemma~\ref{lem:noPE:L}, we conclude that $\wt$ is non-PE again by Theorem~\ref{thm:PE:stack:props}. Expressing $\wt$ in terms of projection pairs, we have that
    \begin{align*}
        \wt = (U_{\cW} - U_{\cW\0} (D_{\cV\0} U_{\cW})) w\pe
        \,.
    \end{align*}
    Since $\wt$ is non-PE, its PE subspace
    must be $\tilde{\cW} := (\RR^q)^{\perp} = \{\, 0 \,\}$. But by Lemma~\ref{lem:PE:L:PEsub} we also know that 
    \begin{align*}
        \{\, 0 \,\} 
        = 
        \tilde{\cW}
        = 
        (U_{\cW} - U_{\cW\0} (D_{\cV\0} U_{\cW})) \RR^{q\pe}
    \end{align*}
    since $w\pe$ is PE. Therefore $U_{\cW} = U_{\cW\0} (D_{\cV\0} U_{\cW})$, implying that $\cW = \Img(U_{\cW}) \subseteq \Img(U_{\cW\0}) = \cW\0$ by definition of the insertion maps.
    To show the reverse inclusion $\cW\0 \subseteq \cW$, observe that by the same arguments as before we have that
    \begin{align*}
        0
        &=
        w - w
        \\&
        =
        (\Pi_{\cW\0}^{\cV\0} - \Pi_{\cW}^{\cV}) w 
        + \Pi_{\cV\0}^{\cW\0} w + (- \Pi_{\cV}^{\cW} w)
        \\&
        =
        (\Pi_{\cW\0}^{\cV\0} - \Pi_{\cW}^{\cV}) (\Pi_{\cW\0}^{\cV\0} + \Pi_{\cV\0}^{\cW\0}) w
        + \Pi_{\cV\0}^{\cW\0} w + (- \Pi_{\cV}^{\cW} w)
        \\&
        =:
        \wt\0
        + (\Pi_{\cW\0}^{\cV\0} - \Pi_{\cW}^{\cV}) \Pi_{\cV\0}^{\cW\0} w
        + \Pi_{\cV\0}^{\cW\0} w + (- \Pi_{\cV}^{\cW} w)
    \end{align*}
    where $\wt\0 = (U_{\cW\0} - U_{\cW} (D_{\cV} U_{\cW\0})) w_1$ is non-PE. Using the fact that the PE subspace of $\wt\0$ is 
    \begin{align*}
        \{\, 0 \,\} 
        =
        \tilde{\cW}\0 
        =
        (U_{\cW\0} - U_{\cW} (D_{\cV} U_{\cW\0})) \RR^{\dim(\cW\0)}
    \end{align*}
    since $w_1$ is PE, we obtain $U_{\cW\0} = U_{\cW} (D_{\cV} U_{\cW\0})$ and thus $\cW\0 = \Img(U_{\cW\0}) \subseteq \Img(U_{\cW}) = \cW$.
\end{proof}

\section{Conditions for Regularity}
\label{sec:conds}

We have presented a geometric notion of regular regressors. We now go deeper to explore sufficient conditions for regularity in terms of signal properties. Arguably the simplest condition arises in the case of vanishing regressors. The following is a consequence of  \cite[Lemma~6.5~(i)]{NARENDRA89}.

\begin{proposition} \label{prop:vanish_noPE}
    Suppose $w(t) \in \RR^q$ satisfies $w \to 0$.
    Then for all $\al \in \RR^q$, $\al\T w$ is not PE.
    \dqed
\end{proposition}

In light of the above, a natural sufficient condition for regularity is to assume that any $\al\T w$ that is not PE is vanishing. The proof follows from linearity of limits.

\begin{proposition} \label{prop:PE:woPE_vanish:reg}
    Suppose $w(t) \in \RR^q$ is bounded. Also suppose that for any $\al \in \RR^q$ we have $\al\T w$ is not PE implies $\al\T w \to 0$.
    Then $w$ is a regular regressor.
    \dqed
\end{proposition}
% \begin{proof}
%     Let $\al_1,\, \al_2 \in \cW^{\star}$ where $\cW^{\star}$ is the non-PE set of $w$. By assumption we have $\al_i\T w \to 0$. Letting $\gamma \in \RR$ we have that $\al := \gamma \al_1 + \al_2$ satisfies $\al\T w = \gamma (\al_1\T w) + (\al_2\T w) \to 0$. But then by Proposition~\ref{prop:vanish_noPE} we have $\al\T w$ is not PE and so $\al \in \cW^{\star}$, showing closure under linear combinations.
% \end{proof}

\subsection{Dealing with Expanding Intervals of Inactivity}

The pathology exposed in Section~\ref{sec:PE:reg:pathology} is that one may have a regressor $w = (w_1, w_2)$ with non-vanishing components that are not PE yet reoccurring excitation may still exist with increasingly long intervals of inactivity. To amend this pathology, we need to ensure that excitation, if persistent, reoccurs at some finite interval. A property that best captures this notion is {\em relative density}, found in \cite[Definition~3, p.~342]{HALE80} when discussing {\em almost periodic} signals.

\begin{definition}%[Relatively Dense] 
    \label{def:PE:RelDense}
    We say $\cT \subseteq [0, \infty)$ is {\em relatively dense} if there exists a $T_d > 0$ such that $\cT \cap [t, t + T_d] \neq \emptyset$ for all $t \in [0, \infty)$.
    \rqed
\end{definition}

Now fix any vector $\al \in \RR^q$ and constants $\beta,\, T > 0$. For a regressor $w(t) \in \RR^q$ define the set
\begin{align*}
    \cT^w&(\al, \beta, T)
    :=
    \\&
    \{\, t \geq 0 ~:~ \al\T \left(
    \frac{1}{T} \int_{t}^{t + T} w(\tau) w\T(\tau) \,d\tau
    \right) \al \geq \beta \| \al \|^2 \,\}
\end{align*}
that records all the times the regressor is excited along $\al$.
Since the length of intervals of inactivity is the main issue encountered with the pathology of Section~\ref{sec:PE:reg:pathology}, we consider the following condition.
\begin{enumerate}[label=(R)]
    \item \label{item:PE:R}
    For each $\al \in \RR^q$, either there exist $\beta,\, T > 0$ such that $\cT^w(\al, \beta, T)$ is relatively dense or for all $\beta,\, T > 0$ we have $\sup \cT^w(\al, \beta, T) < \infty$.
\end{enumerate}
Condition~\ref{item:PE:R} states that either there is an excitation level that reoccurs at some finite interval or that any excitation encountered stops eventually. In particular, reoccurring excitation with increasingly long intervals of inactivity is not permitted.

To ensure that non-vanishing regressors are indeed PE, we will require a continuity assumption on the regressor in addition to condition \ref{item:PE:R}. To see why, consider a smooth pulse train but where the pulses get narrower as time goes on. Clearly such a regressor is non-vanishing, but it is also not PE as narrower pulses mean that there is less excitation available over time. To end up with regressors where shorter and shorter bursts of excitation do not occur, we will limit the variation that regressors can undergo.
This is done by considering the class $\cP_{[0, \infty)}$ but relaxing boundedness of the derivative to a kind of uniform continuity.

\begin{definition} \label{def:upConst}
    We say $\rho \colon [0, \infty) \to \RR^q$ is {\em uniformly piecewise constant} if there exists a set $\cJ \subseteq [0, \infty)$ such that:
    \begin{enumerate}[label=\alph*)]
        \item
        $\rho$ is right continuous;
        \item 
        $\rho$ is constant on every connected interval in $[0, \infty) \setminus \cJ$;
        \item \label{itm:def:upConst:c}
        there exists a constant $\delta > 0$ such that $| t_2 - t_1 | \geq \delta$ for all distinct $t_1,\, t_2 \in \cJ$.
        \rqed
    \end{enumerate}
\end{definition}

\begin{remark}
    The prefix {\em uniformly} in Definition~\ref{def:upConst} is to emphasize that there must be a uniform minimum gap between times of discontinuities in point \ref{itm:def:upConst:c}.
    \tqed
\end{remark}

\begin{definition}%[Uniform Piecewise Continuity]
    We say $w \colon [0, \infty) \to \RR^q$ is {\em uniformly piecewise continuous} if there exists a uniformly piecewise constant function $\rho$ such that $w - \rho$ is uniformly continuous.
    \rqed
\end{definition}

We now show that enforcing intervals of inactivity to have a maximal length and assuming uniform piecewise continuity together ensure that regressors are well-behaved. In particular, these two requirements enforce that scalar regressors that are not PE are in fact vanishing.

\begin{theorem} \label{thm:PE:reg:R_vanish}
    Suppose $w(t) \in \RR^q$ is uniformly piecewise continuous.
    The following are equivalent:
    \begin{enumerate}[label=\alph*)]
        \item \label{itm:thm:PE:reg:R_vanish:R}
        $w$ satisfies \ref{item:PE:R};
        \item \label{itm:thm:PE:reg:R_vanish:vanish}
        for any $\al$ we have $\al\T w$ is not PE implies $\al\T w \to 0$.
        \dqed
    \end{enumerate}
\end{theorem}
\begin{proof}
    We consider $\alpha \neq 0$ as the requirements in both \ref{itm:thm:PE:reg:R_vanish:R} and \ref{itm:thm:PE:reg:R_vanish:vanish} are trivially satisfied when $\alpha = 0$.
    We first show that $\al\T w$ is PE if and only if there exist $\beta,\, T > 0$ such that $\cT^w(\al, \beta, T)$ is relatively dense.

    ($\implies$)
    Suppose $\al\T w$ is PE. 
    By Definition~\ref{def:PE} there exist $\beta\pe,\, T\pe > 0$ such that $\cT^w(\al, \beta\pe \| \al \|^{-2}, T\pe) = [0, \infty)$. Clearly $[0, \infty)$ is relatively dense.

    ($\impliedby$)
    Suppose $\cT^w(\al, \beta, T)$ is relatively dense for some $\beta,\, T > 0$.
    Then there exists a $T_d > 0$ such that for every $t \geq 0$ there exists a time $t_1(t)$ satisfying $t_1 \in \cT^w(\al, \beta, T) \cap [t, t + T_d]$.
    Let $\gamma > 1$ be such that $T_d + T = \gamma T$. Then $[t_1, t_1 + T] \subseteq [t, t + \gamma T]$ and so
    \begin{align*}
        \frac{1}{\gamma T} &\int_{t}^{t + \gamma T} \al\T w(\tau) w\T(\tau) \al \,d\tau
        \\&
        \geq
        \frac{1}{\gamma} \al\T \left(
        \frac{1}{T} \int_{t_1}^{t_1 + T} w(\tau) w\T(\tau) \,d\tau
        \right) \al
        \geq
        \frac{\beta}{\gamma} \| \al \|^2 \cdot 1
    \end{align*}
    for all $t \geq 0$; that is, $\al\T w$ is PE.

    Now we show that \ref{itm:thm:PE:reg:R_vanish:R} and \ref{itm:thm:PE:reg:R_vanish:vanish} are equivalent.

    ($\implies$)
    Suppose $w$ satisfies \ref{item:PE:R}.
    Let $\al \in \RR^q$ be such that $\al\T w$ is not PE. By our work above we have that $\cT^w(\al, \beta, T)$ is not relatively dense for any $\beta,\, T > 0$. By \ref{item:PE:R} this tells us that $\sup \cT^w(\al, \beta, T) < \infty$ for all $\beta,\, T > 0$.
    For the sake of contradiction, suppose $\al\T w \not\to 0$.
    Then for some $\beta\0 > 0$ there exists a sequence $\{\, t_k \,\}_{k = 1}^{\infty}$ such that $t_k \to \infty$ and $| \al\T w(t_k) | \geq 2 \sqrt{\beta\0} \| \al \|$. Since $w$ is uniformly piecewise continuous, so is $\al\T w$. As such, there exists $T\0 > 0$ such that for all $k \in \NN$ we have that $| \al\T w(t) | \geq \sqrt{\beta\0} \| \al \|$ for either all $t \in [t_k, t_k + T\0]$ or all $t \in [t_k - T\0, t_k]$. For every $k \in \NN$, if the previous statement only holds for $t \in [t_k - T\0, t_k]$, replace $t_k$ with $t_k - T\0$ for that $k$. In doing so, we have that
    \begin{align*}
        \al\T \left(
        \frac{1}{T\0} \int_{t_k}^{t_k + T\0} w(\tau) w\T(\tau) \,d\tau
        \right) \al
        \geq
        \beta\0 \| \al \|^2
    \end{align*}
    and so $\{\, t_k \,\}_{k = 1}^{\infty} \subseteq \cT^w(\al, \beta\0, T\0)$. Since $t_k \to \infty$, this contradicts the fact that $\sup \cT^w(\al, \beta\0, T\0) < \infty$.

    ($\impliedby$)
    Suppose $\al\T w$ is not PE implies $\al\T w \to 0$.
    If $\al\T w$ is PE then our work at the start of the proof of this theorem showed that there exist $\beta,\, T > 0$ such that $\cT^w(\al, \beta, T)$ is relatively dense; that is, \ref{item:PE:R} is satisfied.
    So it suffices to consider the case when $\al\T w$ is not PE.
    In this case, for every $T > 0$ we have
    \begin{align*}
        \lim_{t \to \infty} \al\T \left(
        \frac{1}{T} \int_{t}^{t + T} w(\tau) w\T(\tau) \,d\tau
        \right) \al
        =
        0
        \,,
    \end{align*}
    from which it follows that $\sup \cT^w(\al, \beta, T) < \infty$ for every $\beta,\, T > 0$; that is, \ref{item:PE:R} is satisfied.
\end{proof}

In light of Theorem~\ref{thm:PE:reg:R_vanish}, a natural class of regressors to consider that are void of the pathology in Section~\ref{sec:PE:reg:pathology} are those whose components that are not PE do indeed vanish.

\begin{assumption} \label{as:woPE_vanish}
    The regressor $w(t) \in \RR^q$ is bounded, uniformly piecewise continuous, and for any $\al \in \RR^q$ we have $\al\T w$ is not PE implies that $\al\T w \to 0$.
    \rqed
\end{assumption}

The following theorem summarizes our findings about the equivalence between vanishing regressors and regularity, and is proved by combining Proposition~\ref{prop:PE:woPE_vanish:reg} with Theorem~\ref{thm:PE:reg:R_vanish}.

\begin{theorem}
    Suppose $w(t) \in \RR^q$ is uniformly piecewise continuous.
    The following are equivalent:
    \begin{enumerate}[label=\alph*)]
        \item \label{itm:thrm:A1:a}
        $w$ is regular and satisfies \ref{item:PE:R};
        \item \label{itm:thrm:A1:b}
        $w$ satisfies Assumption~\ref{as:woPE_vanish}.
        \dqed
    \end{enumerate}
\end{theorem}

This fundamental result states that to reconcile the geometric notion of regularity in \ref{itm:thrm:A1:a} and the property that non-PE regressors are vanishing in \ref{itm:thrm:A1:b}, one must impose condition  \ref{item:PE:R}: a 
% manifestly 
reasonable requirement to eliminate pathological behaviour.

\section{New Problems in Adaptive Control}
\label{sec:probs}

The primary goal in adaptive control is to design an adaptive law for an adaptive parameter $\psih(t) \in \RR^q$ so that an error, such as $e := w\T \psih - d = w\T (\psih - \psi)$, vanishes.
Whether this goal also achieves $\psih \to \psi$ is known to depend on the excitation of the regressor. If the regressor is not PE then the requirement of $e$ vanishing asymptotically may still allow for infinitely many possible values for $\psih$.

\begin{lemma} \label{lem:ALB:e_implies_Wpr}
    Suppose $\psih \in \RR^q$ is constant.
    Let $\cW^{\star}$ denote the non-PE set of $w$.
    If $(\psih - \psi)\T w \to 0$ then $\psih - \psi \in \cW^{\star}$.
    \dqed
\end{lemma}
\begin{proof}
    By Proposition~\ref{prop:vanish_noPE} we know $(\psih - \psi)\T w$ is not PE. By definition of $\cW^{\star}$ it follows that $\psih - \psi \in \cW^{\star}$.
\end{proof}

The significance of this simple lemma is that when regressors are not PE, additional degrees of freedom become available to an adaptive control designer to select a desirable estimate $\psih$ while still achieving $e \to 0$. Moreover, if $w$ is a regular regressor then these degrees of freedom correspond to the affine set $\psi + \cW^{\perp}$ (since $\cW^{\star} = \cW^{\perp}$ is a subspace). The converse of Lemma~\ref{lem:ALB:e_implies_Wpr} holds true if the regressor satisfies Assumption~\ref{as:woPE_vanish}. As such, we can replace the regulation requirement $e \to 0$ with the geometric constraint $\psih - \psi \in \cW^{\perp}$. By managing the degrees of freedom that arise when regressors are not PE, one can pose new problems in adaptive control.

\subsection{Prior Knowledge Retention}

The problem is to find some $\psih$ that is as close as possible to a nominal parameter value $\psi\0$ while maintaining error regulation. This nominal value may represent the designer's prior knowledge of the parameters.

\subsection{Adaptive Load Balancing}

The adaptive load balancing problem is to optimize a cost over the set $\psi + \cW^{\perp}$. When multiple regressors  are involved, a designer can trade off (e.g.) energy consumption, reliability, latency, or memory usage associated to each regressor source.

\subsection{Optimal Steady-State Regulation (OSSR)}

The OSSR problem regards a class of adaptive load balancing problems in which a cost is expressed in terms of the input components in steady-state, rather than simply the parameter estimates \cite{HAFEZ24_TAC}.

\subsection{New Approach to Robustness}

The robust adaptive control problem is to design an adaptive law for $\psih$ that achieves error regulation and remains robust to noise, regardless of whether the regressor is PE or not. 
The non-PE set exposes where non-robustness can occur, resulting in the $\mu$-modification \cite{MEJIA23_SCL, MEJIA24_TAC, MEJIA23_CDC, MEJIA23_ARC}, and thereby avoiding ad hoc methods to handle robustness.

\section{Conclusion}

This paper proposes a new notion of regularity of regressors based on the observation that the PE condition permits unreasonable behaviour.
The most significant finding is that regular regressors are amenable to a PE decomposition showing that all excitation is confined to a PE subspace. Regularity is also related to the commonsense idea that non-PE regressors should be vanishing. Finally, the paper lays the foundation to formulate new problems in adaptive control.

% References
\bibliographystyle{ieeetr}
\bibliography{main}

\end{document}

%% file: commands.tex
\newcommand{\NN}{\mathbb{N}}
\newcommand{\RR}{\mathbb{R}}

\newcommand{\cJ}{\mathcal{J}}

\newcommand{\cP}{\mathcal{P}}

\newcommand{\cT}{\mathcal{T}}

\newcommand{\cV}{\mathcal{V}}
\newcommand{\cW}{\mathcal{W}}

\newcommand{\pe}{_\text{pe}}

\newcommand{\pr}{_{\perp}}

\newcommand{\x}{\times}
\newcommand{\T}{^{\intercal}}

\newcommand{\0}{_{\circ}}

\newcommand{\al}{\alpha}

\newcommand{\psih}{\hat{\psi}}

\newcommand{\wt}{\tilde{w}}

\newcommand\tqed{\leavevmode\unskip\penalty9999 \hbox{}\nobreak\hfill\quad\hbox{$\triangleleft$}}
\newcommand{\dqed}{\leavevmode\unskip\penalty9999 \hbox{}\nobreak\hfill\quad\hbox{$\diamond$}}
\newcommand{\rqed}{\leavevmode\unskip\penalty9999 \hbox{}\nobreak\hfill\quad\hbox{$\triangleright$}}

\renewcommand{\QED}{\QEDopen}

\DeclareMathOperator{\Img}{Im}
\DeclareMathOperator{\Ker}{Ker}